\documentclass[11pt]{article} 

\usepackage[utf8]{inputenc} 

\usepackage{geometry} 
\usepackage{graphicx} 
\usepackage{color}

\definecolor{Red}{rgb}{1,0,0}
\definecolor{Blue}{rgb}{0,0,1}
\definecolor{Olive}{rgb}{0.41,0.55,0.13}
\definecolor{Green}{rgb}{0,1,0}
\definecolor{MGreen}{rgb}{0,0.8,0}
\definecolor{DGreen}{rgb}{0,0.55,0}
\definecolor{Yellow}{rgb}{1,1,0}
\definecolor{Cyan}{rgb}{0,1,1}
\definecolor{Magenta}{rgb}{1,0,1}
\definecolor{Orange}{rgb}{1,.5,0}
\definecolor{Violet}{rgb}{.5,0,.5}
\definecolor{Purple}{rgb}{.75,0,.25}
\definecolor{Brown}{rgb}{.75,.5,.25}
\definecolor{Grey}{rgb}{.5,.5,.5}

\usepackage{booktabs} 
\usepackage{array} 
\usepackage{paralist} 
\usepackage{verbatim} 
\usepackage{subfigure} 
\usepackage{amsmath,amssymb,amsthm,mathrsfs}
\usepackage[colorlinks]{hyperref}
\usepackage{blkarray}
\usepackage{bbm}
\usepackage{mathtools,xparse}
\usepackage{tikz}  
\usetikzlibrary{cd}

\DeclareMathOperator{\supp}{supp}

\usepackage{fancyhdr} 
\theoremstyle{plain}

\newtheorem{theorem}{Theorem}

\newtheorem{corollary}{Corollary}[theorem]

\newtheorem{lemma}{Lemma}

\newtheorem{theorem*}{Theorem}   
\newtheorem{lemma*}{Lemma} 
\newtheorem{corollary*}{Corollary} 
\newtheorem*{remark*}{Remark}
\newtheorem{example}{Example}

\newtheorem{remark}{Remark}[section]

\newlength{\widebarargwidth}
\newlength{\widebarargheight}
\newlength{\widebarargdepth}

\theoremstyle{definition}
\newtheorem{definition}{Definition}

\def\cG{{\cal G}}

\def\cK{{\cal K}}

\def\cP{{\cal P}}

\newcommand{\real}{\ensuremath{\mathbb{R}}}
\newcommand{\defn}{\ensuremath{:  =}}

\DeclarePairedDelimiter{\abs}{\lvert}{\rvert}
\DeclarePairedDelimiter{\norm}{\lVert}{\rVert}

\usepackage{soul}  
\DeclareMathOperator{\vol}{Vol}

\begin{document}

\begin{center}

	{\bf{\LARGE{Dual Loomis-Whitney inequalities via information theory}}}

	\vspace*{.25in}

	\begin{tabular}{ccc}
		{\large{Jing Hao}} & \hspace*{.75in} & {\large{Varun Jog}} \\ 
		{\large{\texttt{jing.hao@wisc.edu}}} & & {\large{\texttt{vjog@wisc.edu}}} \vspace{.2in}
		\\
		Department of Mathematics & \hspace{.2in} & Department of Electrical \& Computer Engineering \\
		University of Wisconsin - Madison && University of Wisconsin - Madison \\ 
	\end{tabular}

	\vspace*{.2in}

	January 2019

	\vspace*{.2in}

\end{center}

\abstract{We establish lower bounds on the volume and the surface area of a geometric body using the size of its slices along different directions. In the first part of the paper, we derive volume bounds for convex bodies using generalized subadditivity properties of entropy  combined with entropy bounds for log-concave random variables. In the second part, we investigate a new notion of Fisher information which we call the $L_1$-Fisher information, and show that certain superadditivity properties of the $L_1$-Fisher information lead to lower bounds for the surface areas of polyconvex sets in terms of its slices.}

\section{Introduction}

Tomography concerns reconstructing a probability density by synthesizing data collected along sections (or slices) of that density, and is a problem of great significance in applied mathematics. Some popular applications of tomography in the field of medical imaging are computed tomography (CT), magnetic resonance imaging (MRI), and positron emission tomography (PET). In each of these, sectional data is obtained in a non-invasive manner using penetrating waves, and images are generated using tomographic reconstruction algorithms. \emph{Geometric tomography} is a term coined by Gardner \cite{Gar95} to describe an area of mathematics that deals with the retrieval of information about a geometric object from data about its sections, projections, or both. Gardner notes that the term \emph{geometric} is deliberately vague, since it may be used to describe study convex sets or polytopes as well as more general shapes such as star-shaped bodies, compact sets, or even Borel sets. 

An important problem in geometric tomography is estimating the size of set using lower dimensional sections or projections. Here, projection of a geometric object refers to its shadow, or orthogonal projection, as opposed to the marginal of a probability density. As detailed in Campi and Gronchi \cite{CamGro11}, this problem is relevant in a variety settings ranging from the microscopic study of biological tissues \cite{WulEtal04, WulEtal04b}, to the study of fluid inclusions in minerals \cite{SheEtal85, BakDia06}, and to reconstructing the shapes of celestial bodies \cite{ConOst84, OstCon84}. Various geometric inequalities provide bounds on the sizes of sets using lower dimensional data pertaining to projections and slices of sets. The ``size" of a set often refers to its volume, but it may also refer to more general geometric properties such as surface area or mean width. A canonical example of an inequality that bounds the volume of set using its orthogonal projections is the Loomis-Whitney inequality \cite{LooWhi49}. This inequality states that 
for any Borel measurable set $K \subseteq \real^n$,
\begin{equation}\label{eq: lw}
V_n(K) \leq \left(\prod_{i=1}^n V_{n-1}(P_{e_i^\perp} K)\right)^{\frac{1}{n-1}}.
\end{equation}
Equality holds in \eqref{eq: lw} if and only if $K$ is a box with sides parallel to the coordinate axes. The Loomis-Whitney inequality has been generalized and strengthened in a number of ways. Burago and Zalgaller \cite{BurZal13} proved a version \eqref{eq: lw} that considers projections of $K$ on to all $m$-dimensional spaces spanned by $\{e_1, \dots, e_n\}$. Bollobas and Thomason \cite{BolTho95} proved the Box Theorem which states that for every  Borel set $K \subseteq \real^n$, there exists a box $B$ such that $V_n(B) = V_n(K)$ and $V_m(P_S B) \leq V_m(P_S K)$ for every $m$-dimensional coordinate subspace $S$. Ball \cite{Ball91} showed that the Loomis-Whitney inequality is closely related to the Brascamp-Lieb inequality \cite{Ball01, BenEtal08} from functional analysis, and generalized it to projections along subspaces that satisfy a certain condition. Inequality \eqref{eq: lw} also has deep connections to additive combinatorics and information theory.  Some of these connections have been explored in Balister and Bollobas \cite{BalBol12} and Gyarmati et al. \cite{GyaEtal10}, and Madiman and Tetali \cite{MadTet10}. 

A number of geometric inequalities also provide upper bounds for the surface area of a set using projections. Naturally, it is necessary to make some assumptions for such results, since one can easily conjure sets that have small projections while having a large surface area.  Betke and McMullen \cite{BetMul83, CamGro11} proved that for compact convex bodies, 
\begin{equation}\label{eq: lw2}
V_{n-1}(\partial K) \leq 2 \sum_{i=1}^n V_{n-1} (P_{e_i^\perp} K).
\end{equation}
Motivated by inequalities \eqref{eq: lw} and \eqref{eq: lw2}, Campi and Gronchi \cite{CamGro11} investigated upper bounds for \emph{intrinsic volumes} \cite{Sch14} of compact convex sets.

Inequalities \eqref{eq: lw} and \eqref{eq: lw2} provide upper bounds, and a natural question of interest is developing analogous lower bounds. Lower bounds are obtained via \emph{reverse} Loomis-Whitney inequalities or \emph{dual} Loomis-Whitney inequalities. The former uses projection information whereas the latter uses slice information, often along the coordinate axes. A canonical example of a dual Loomis-Whitney inequality is Meyer's inequality \cite{Mey88}, which states that for a compact convex set $K \subseteq \real^n$, the following lower bound holds:
\begin{equation}\label{eq: lw3}
V_n(K) \geq \left( \frac{n!}{n^n} \prod_{i=1}^n V_{n-1}(K \cap e_i^\perp) \right)^{\frac{1}{n-1}},
\end{equation}
with equality if and only if $K$ is a regular crosspolytope. Betke and McMullen \cite{BetMul83, CamGro11} established a reverse Loomis-Whitney type inequality for surface areas of compact convex sets:
\begin{equation}\label{eq: lw4}
V_{n-1}(\partial K)^2 \geq 4 \sum_{i=1}^n V_{n-1}(P_{e_i^\perp} K)^2.
\end{equation}
Campi et al. \cite{CamEtal16} extended inequalities \eqref{eq: lw3} and \eqref{eq: lw4} for intrinsic volumes of certain convex sets.

Our goal in this paper is to develop lower bounds on volumes and surface areas of geometric bodies that are most closely related to dual Loomis-Whitney inequalities; i.e., inequalities that use slice-based information. The primary mathematical tools we use are entropy and information inequalities; namely, the Brascamp-Lieb inequality, entropy bounds for log-concave random variables, and superadditivity properties of a suitable notion of Fisher information. Using information theoretic tools allows our results to be quite general. For example, our volume bounds rely on maximal slices parallel to a set of subspaces, and are valid for very general choice of subspaces. Our surface area bounds  are valid for polyconvex sets, which are of finite unions of compact convex sets. The drawback of using information theoretic strategies is that the resulting bounds are not always tight; i.e., equality may not achieved by any geometric body. However, we show that in some cases our bounds are asymptotically tight as the dimension $n$ tends to infinity, thus partly mitigating the drawbacks. Our main contributions are as follows:
\begin{itemize}
\item
\emph{Volume lower bounds:} In Theorem \ref{thm: vol}, we establish a new lower bound on the volume of a compact convex set in terms of the size of its slices. Just as Ball \cite{Ball91} extended the Loomis-Whitney inequality to projections in more general subspaces, our inequality also allows for slices parallel to subspaces that are not necessarily $e_i^\perp$. Another distinguishing feature of this bound is that unlike classical dual Loomis-Whitney inequalities, the lower bound is in terms of \emph{maximal slices}; i.e. the largest slice parallel to a given subspace. The key ideas we use are the Brascamp-Lieb inequality and certain entropy bounds for log-concave random variables.

\item 
\emph{Surface area lower bounds: }  Theorem \ref{thm: surface_area} contains our main result that provides lower bounds for surface areas. Unlike the volume bounds, the surface area bounds are valid for the  larger class of polyconvex sets, which consists of finite unions of compact, convex sets. Moreover, the surface area lower bound is not simply in terms of the maximal slice; instead, this bound uses all available slices along a particular hyperplane. As in the volume bounds, the slices used may be parallel to general $(n-1)$-dimensional subspaces, and not just $e_i^\perp$. The key idea is motivated by a superadditivity property of Fisher information established in Carlen \cite{Car91}. Instead of classical Fisher information, we develop superadditivity properties for a new notion of Fisher information which we call the $L_1$-Fisher information. This superadditivity property when restricted to uniform distributions over convex bodies yields the lower bound in Theorem \ref{thm: surface_area}.

\end{itemize}

The paper is structured as follows. In Sections \ref{section: volume} we state and prove our volume lower bound, and in Section \ref{section: surface} we state and prove our surface area bound. We conclude with some open problems and discussions in Section \ref{section: end}.

\paragraph{Notation: } For $n \geq 1$, let $[n]$ denote the set $\{1, 2, \dots, n\}$. For $K \subseteq \real^n$ and any subspace $E \subseteq \real^n$, the orthogonal projection of $K$ on $E$ is denoted by $P_E K$. The standard basis vectors in $\real^n$ are denoted by $\{e_1, e_2, \dots, e_n\}$. We use the notation $V_r$ to denote the volume functional in $\real^r$. The boundary of $K$ is denoted by $\partial K$, and its surface area is denoted by $V_{n-1}(\partial K)$. For a random variable $X$ taking values in $\real^n$, the marginal of $X$ along a subspace $E$ is denoted by $P_E X$. In this paper, we shall consider random variables with bounded variances and whose densities lie in the convex set
$\{f| \int_{\real^n} f(x) \log (1+f(x)) < \infty$. The differential entropy of such random variables is well-defined, and is given by 
$$h(X) = -\int_{\real^n}p_X(x) \log p_X(x) dx,$$
where $X \sim p_X$ is an $\real^n$-valued random variable. The Fisher information of a random variable $X$ with a differentiable density $p_X$ is given by
$$I(X) = \int_{\real^n} \norm{\nabla \log p_X(x)}^2 p_X(x) dx.$$

\section{Volume bounds}\label{section: volume}
The connection between functional/information theoretic inequalities and geometric inequalities is well-known. In particular, the Brascamp-Lieb inequality has found  several applications in geometry as detailed in Ball \cite{Ball01}. In the following section we briefly discuss the Brascamp-Lieb inequality and its relation to volume inequalities.
  
\subsection{Background on the Brascamp-Lieb inequality}
We shall use the the information theoretic form of the Brascamp-Lieb inequality, as found in Carlen et al. \cite{CarCor09}:
\begin{theorem}\label{thm: BL}[Brascamp-Lieb inequality]
Let $X$ be random variable taking values in $\real^n$. Let $E_1, E_2, \dots, E_m \subseteq \real^n $ be subspaces and $c_1, c_2, \dots, c_m > 0$ be constants. Define 
\begin{equation}
M = \sup_{X} h(X) - \sum_{j=1}^m c_j h(P_{E_i} X,
\end{equation}
and
\begin{equation}
M_g = \sup_{X \in \cG} h(X) - \sum_{j=1}^m c_j h(P_{E_i} X),
\end{equation}
where $\cG$ is the set of all Gaussian random variables taking values in $\real^n$. Then
$M = M_g$, and $M_g$ (and therefore $M$) is finite if and only if $\sum_{i=1}^m r_i c_i = n$ and for all subspaces $V \subseteq \real^n$, we have $\dim(V) \leq \sum_{i=1}^n \dim(P_{E_i} V) c_i$.
\end{theorem}
Throughout this paper, we assume that $E_i$ and $c_i$ are such that $M < \infty$. As detailed in Bennett et al. \cite{BenEtal08}, the Brascamp-Lieb inequality generalizes many popular inequalities such as Holder's inequality, Young's convolution inequality, and the Loomis-Whitney inequality. In particular, Ball \cite{Ball91} showed that the standard Loomis-Whitney inequality in \eqref{eq: lw} could be extended to settings where projections are obtained on more general subspaces:

\begin{theorem}[Ball \cite{Ball91}]\label{thm: BL_vol}
Let $K$ be a closed and bounded set in $\real^n$. Let $E_i$ and $c_i$ for $i \in [m]$, and $M_g$ be as in Theorem \ref{thm: BL}. Let $P_{E_i}K$ be the projection of $K$ on to the subspace $E_i$, for $i \in [m]$.  Let the dimension of $E_i$ be $r_i$ for $ \in [m]$. Then the volume of $K$ may be upper-bounded as follows:
\begin{equation}
V_n(K) \leq e^{M_g} \prod_{i=1}^m V_{r_i}(P_{E_i}K)^{c_i}.
\end{equation}
\end{theorem}
Since we shall be using a similar idea in Section \ref{section: volume}, we include a proof for completeness.
\begin{proof}
Consider a random variable $X$ that is uniformly distributed on $K$; i.e. $X \sim p_X = \text{Unif}(K)$. Let $P_{E_i}X$ denote the random variable obtained by projecting $X$ on $E_i$, or equivalently the marginal of $X$ in subspace $E_i$. Naturally, $\supp(P_{E_i}X) \subseteq P_{E_i}K$, and thus
\begin{equation}
h(P_{E_i}X) \leq \log V_{r_i} (P_{E_i}K), \quad \text{ for $i \in [n]$.}
\end{equation}
Substituting these inequalities in the Brascamp-Lieb inequality for $X$, we obtain
\begin{align}
h(X) = \log V_n(K) \leq \sum_{j=1}^m c_j \log V_{r_i}(P_{E_i} K) + M_g.
\end{align}
Exponentiating both sides concludes the proof.
\end{proof}
To show that the Loomis-Whitney inequality is implied by Theorem \ref{thm: BL_vol}, we set $E_i = e_i^\perp$, $c_i = n/(n-1)$ for $i \in [n]$, and use Szasz's inequality or other tools from linear algebra \cite{BecBel12} to show that the supremum below evaluates to 1:
\begin{align*}
e^{M_g} = \sup_{K \succeq 0} \frac{\det{K}}{\prod_{i=1}^n \det{K_i}^{\frac{1}{n}}}.
\end{align*}
In general, Ball \cite{Ball91} showed that if the $E_i$ and $c_i$ satisfy what is called John's condition; i.e. $\sum_{i=1}^m c_iP_{E_i}x = x$ for all $x \in \real^n$, then $M_g = 0$.

\subsection{Volume bounds using slices}
Providing \emph{lower bounds} for volumes in terms of projections requires making additional assumptions on the set $K$. A simple counterexample is the $(n-1)$ dimensional sphere (shell), which can have arbitrarily large projections in lower dimensional subspaces, but has 0 volume. Even for convex $K$, providing lower bounds using a finite number of projections fails. For example, given a finite collection of subspaces, we may consider any convex set supported on a random $(n-1)$ dimensional subspace of $\real^n$ which will have (with high probability) non-zero projections on all subspaces in the collection. Clearly, such a set has volume 0. Therefore, it makes sense to obtain lower bounds on volumes using \emph{slices} instead of projections, as in Meyer's inequality \eqref{eq: lw3}. 

Given a subspace $E_i$, the slice parallel to $E_i^\perp$ is not unambiguously defined as it depends on translations of $E_i^\perp$. For this reason we consider the \emph{maximal slice}; i.e. the largest slice parallel to a given subspace. Note that although Meyer's inequality \eqref{eq: lw3} is not stated in terms of maximal slices, it remains valid even if the right hand side of inequality \eqref{eq: lw3} is replaced by maximal slices parallel to $e_i^\perp$. This is because one can always choose the origin of the coordinate system such that the largest slice parallel to $e_i^\perp$ is $K \cap e_i^\perp$. However, when subspaces are in a more general orientation, it is not always possible to select the origin that simultaneously maximizes the slices along all subspaces. Our main result is the following:

\begin{theorem}\label{thm: vol}
Let $K$ be a compact convex body in $\real^n$. For $j \in [m]$, let $E_j \subseteq \real^n$ be subspaces with dimensions $r_j$, and $c_j > 0$ be constants. Let $S_{\max}(j)$ be the largest slice of $K$ by a subspace orthogonal to $E_j$; i.e.,
\begin{equation}
S_{\max}(j) = \sup_{t \in E_j} V_{n-r_j}(K \cap (E_j^\perp + t)).
\end{equation}
Then the following inequality holds:
\begin{equation}
V_n(K) \geq \left(\frac{\prod_{j=1}^m S_{\max}(j)^{c_j}}{e^{n+M_g}} \right)^{1/(C-1)},
\end{equation}
where $C = \sum_{j=1}^m c_j$, and $M_g$ is the Brascamp-Lieb constant corresponding to $\{E_i, c_i\}_{i \in [m]}$. 
\end{theorem}
\begin{proof}
There are two main components in the proof. First, let $X$ be a random variable that is uniformly distributed on $K$. The Brascamp-Lieb inequality yields the bound
\begin{align*}
h(X) \leq \sum_{j=1}^m c_j h(P_{E_j} X) + M_g.
\end{align*}
When deriving upper bounds on volume, we employ the upper bound $h(P_{E_i}X) \leq \log V_{r_i}(P_{E_i}K)$. Here, we employ a slightly different strategy. Note that $X$, being a uniform distribution on a convex set, is a log-concave random variable. Thus, any lower dimensional marginal of $X$ is also log-concave \cite{SauWel14}. Furthermore, the entropy of a log-concave random variable is tightly controlled by the maximum value of its density. For a log-concave random variable $Z$ taking values in $\real^n$ and distributed as $p_Z$, it was shown in Bobkov and Madiman \cite{BobMad11} that
\begin{align*}
\frac{1}{n}\log \frac{1}{\norm{p_Z}_\infty} \leq \frac{h(Z)}{n} \leq \frac{1}{n}\log \frac{1}{\norm{p_Z}_\infty} + 1,
\end{align*}
where $\norm{p_Z}_\infty$ is the largest value of the probability density $p_Z$.
Define $Z_i \defn P_{E_i} X$. The key point to note is that $\norm{p_{Z_i}}_\infty$ is given by the size of the largest slice parallel to $E_i^\perp$, normalized by $V_n(K)$; i.e., $\norm{p_{Z_i}}_\infty = \frac{S_{\max}(i)}{V_n(K)}.$
Thus, for $i \in [m]$, 
\begin{align*}
h(Z_i) \leq r_i + \log \frac{1}{\norm{p_{Z_i}}_\infty} = r_i + \log \frac{V_n(K)}{S_{\max}(i)}.
\end{align*}
Substituting this in the Brascamp-Lieb bound, we obtain
\begin{align*}
\log V_n(K) &\leq \sum_{j=1}^m \left( c_j r_j + c_j\log \frac{V_n(K)}{S_{\max}(j)}\right) + M_g\\
&= n + C \log V_n(K) - \sum_{j=1}^m c_j \log S_{\max}(j) + M_g.
\end{align*}
Note that $\sum_{j=1}^m  c_j > \sum_{j=1}^m c_j (r_j/n) = 1,$ and thus we may rearrange and exponentiate to obtain
\begin{align*}
V_n(K) \geq \left(\frac{\prod_{j=1}^m  S_{\max}(j)^{c_j}}{e^{n+M_g}}\right)^{\frac{1}{C-1}}.
\end{align*}
\end{proof}

It is instructive to compare Meyer's inequality to the bound obtained using Theorem \ref{thm: vol} for the same choice of parameters. Substituting $M_g = 0$, $E_i = e_i$, and $c_i =1$, Theorem \ref{thm: vol} gives the bound
\begin{align*}
V_n(K) \geq \left(\frac{\prod_{j=1}^m  S_{\max}(j)}{e^{n}}\right)^{\frac{1}{n-1}}.
\end{align*}
To compare with Meyer's inequality \eqref{eq: lw3}, first assume that the origin of the coordinate plane is selected such that the intersection with $e_i^\perp$ corresponds to the maximal slice along $e_i^\perp$. With such a choice, we may simply compare the constants in the two inequalities. Observe that  
\begin{align*}
\left(\frac{n!}{n^{n-1}}\right)^{\frac{1}{n-1}} \leq \left(\frac{1}{e^n}\right)^{\frac{1}{n-1}},
\end{align*}
and thus Meyer's inequality \eqref{eq: lw3} yields a tighter bound. However,  Sterling's approximation implies that for large enough $n$ the two constants are approximately the same. Thus, Theorem \ref{thm: vol} yields an asymptotically tight result. 

Note that if the slices are not aligned along the coordinate axes, or if the slices are in larger dimensions, then Meyer's inequality \eqref{eq: lw3} is not applicable but Theorem \ref{thm: vol} continues to yield valid inequalities. An important special case is when there are more than $n$ directions along which slices are available. If $u_1, u_2, \dots, u_m$ are unit vectors and constants $c_1, c_2, \dots, c_m$ satisfy John's condition \cite{Ball91}; i.e., $\sum_{j=1}^m c_j P_{u_j} (x) = x$ for all $x \in \real^n$, then Theorem \ref{thm: vol} yields the bound
\begin{align}\label{eq: john_vol}
V_n(K) \geq \left(\frac{\prod_{j=1}^m  S_{\max}(j)^{c_j}}{e^{n}}\right)^{\frac{1}{n-1}},
\end{align}
where $S_{\max}(j)$ is the size of the largest slice by a hyperplane perpendicular to $u_j$. Note that the bound from Theorem \ref{thm: BL_vol} in this case is
\begin{align*}
V_n(K) \leq \prod_{j=1}^m V_{n-1}(P_{u_j^\perp}K)^{c_j},
\end{align*}
which may be compared with inequality \ref{eq: john_vol} by observing  $S_{\max}(j) \leq V_{n-1}(P_{u_j^\perp}K)$.

\section{Surface area bounds}\label{section: surface}

The information theoretic quantities of entropy and Fisher information are closely connected to the geometric quantities of volume and surface area, respectively. Surface area of $K \subseteq \real^n$ is defined as
\begin{equation}
V_{n-1}(\partial K) = \lim_{\epsilon \to 0} \frac{V_n(K \oplus \epsilon B_n) - V_n(K)}{\epsilon},
\end{equation}
where $B_n$ is the Euclidean ball in $\real^n$ with unit radius and $\oplus$ refers to the Minkowski sum. The Fisher information of a random variable $X$ satisfies a similar relation, 
\begin{equation}
I(X) = \lim_{\epsilon \to 0} \frac{h(X+\sqrt \epsilon Z) - h(Z)}{\epsilon},
\end{equation}
where $Z$ is a standard Gaussian random variable that is independent of $X$. Other well-known connections include the relation between the entropy of a random variable and the volume of its typical set \cite{CovTho12}, isoperimetric inequalities concerning Euclidean balls and Gaussian distributions, and the observed similarity between the Brunn-Minkowski inequality and the entropy power inequality \cite{DemEtal91}. In Section \ref{section: volume}, we used subadditivity of entropy as given by the Brascamp-Lieb inequality to develop volume bounds. To develop surface area bounds, it seems natural to use Fisher information inequalities and adapt them to geometric problems. In the following subsection, we discuss relevant Fisher-information inequalities.

\subsection{Superadditivity of Fisher information}

The Brascamp-Lieb subadditivity of entropy has a direct analog noted in \cite{CarCor09}. We focus on the case when $\{u_j\}$ and constants $\{c_j\}$ for $j \in [m]$ satisfy John's condition. The authors in \cite{CarCor09} provide an alternate proof to the Brascamp-Lieb inequality in this case by 
by first showing a superadditive property of Fisher information, which states that
\begin{equation}\label{eq: geom_BL2}
I(X) \geq \sum_{j=1}^m c_j I(P_{u_j} X).
\end{equation}
The Brasamp-Lieb inequality follows by integrating inequality \eqref{eq: geom_BL2} using the following identity that holds for all random variables $X$ taking values in $\real^n$:
\begin{equation}
h(X) = \frac{n}{2} \log 2\pi e - \int_{t=0}^\infty \left(I(X_t) - \frac{n}{1+t} \right) dt,
\end{equation}
where $X_t = X + \sqrt t Z$ for a standard normal random variable $Z$ that is independent of $X$. If $u_i = e_i$ and $c_i = 1$ for $i \in [n]$, then inequality \eqref{eq: geom_BL2} reduces to the superadditivity of Fisher information:
\begin{equation}\label{eq: l2_super}
I(X) \geq \sum_{i=1}^n I(X_i),
\end{equation}
where $X = (X_1, \dots, X_n)$. 

In Section \ref{section: volume}, we directly used the entropic Brascamp-Lieb inequality on random variables uniformly distributed over suitable sets $K \subseteq \real^n$. It is tempting to use inequality \eqref{eq: geom_BL2} to derive surface area bounds for geometric bodies. Unfortunately, directly substituting $X$ to be uniform over $K \subseteq \real^n$ in inequality \eqref{eq: geom_BL2} does not lead to any useful bounds. This is because the left hand side, namely $I(X)$, is $+\infty$ since the density of $X$ is not differentiable. Thus, it is necessary to modify inequality \eqref{eq: geom_BL2} before we can apply it to geometric problems. A classical result concerning superadditivity of Fisher information-like quantities is provided in Carlen \cite{Car91}:
\begin{theorem}[Theorem 2, \cite{Car91}]\label{thm: Car91}
For $p \in [1, \infty)$, let $f: \real^m \times \real^n \to \real$ be a function in $L^p(\real^m) \otimes W^{1,p}(\real^n)$. Define the marginal map $M$ as
\begin{equation}
G(y) = \left(\int_{\real^m} \abs{f(x,y)}^p dx)\right)^{1/p},
\end{equation}
denoted by $Mf = G$. Then the following inequality holds:
\begin{equation}\label{eq: Car91}
\int_{\real^n} \abs{ \nabla_y G(y) }^p dy \leq \int_{\real^m}\int_{\real^n} \abs{\nabla_y f(x,y)}^p dxdy.
\end{equation}
\end{theorem}
Carlen \cite{Car91} also established the (weak) differentiability of $G$ and the continuity of $M$ prior to proving Theorem \ref{thm: Car91}, so the derivatives in its statement are well-defined. The notion of Fisher information we wish to use is essentially identical to the case of $p=1$ in Theorem \ref{thm: Car91}. However, since our goal is to use this result for uniform densities over compact sets, we cannot directly use Theorem \ref{thm: Car91}, since such densities do not satisfy the required assumptions. In particular, the (weak) partial derivatives of conditional densities are defined in terms of Dirac delta distributions which not lie in the Sobolev space $W^{1,1}(\real^n)$. To get around this, we redefine the $p=1$ case as follows:

\begin{definition}\label{def: l1_fisher}
	Let $X=(X_1, \ldots, X_n)$ be a random vector on $\mathbb{R}^n$ and $f_X(\cdot)$ be its density function.  For any unit vector $u \in \mathbb{R}^n$, define
	\[
		I_1(X)_u \defn 	\lim_{\epsilon \rightarrow 0^+} 
		 		\int_{\mathbb{R}} \frac{|f_X(x) - f_X(x - \epsilon u)|}{\epsilon} dx,
	\]
	given that the limit exists. Define the $L_1$-Fisher information of $X$ as
	\[
		I_1(X) \defn \sum_{i=1}^{n} I_1(X)_{e_i},
	\]
	given that the right hand side is well-defined.  In particular, when $X$ is a real-valued random variable,
	\[
		I_1(X) = \lim_{\epsilon \rightarrow 0^+} \int_\mathbb{R} \frac{|f_X(x) - f_X(x- \epsilon)|}{\epsilon}dx.
	\]

	\end{definition}

Our new definition is motivated by observing that Theorem \ref{thm: Car91} is essentially a data processing result for $\phi$-divergences, and specializing it to the total variation divergence yields our definition. To see this, consider real-valued random variables $X$ and $Y$ with a joint density $\tilde f(x,y)$. Let the marginal of $Y$ on $\real$ be $\tilde G(\cdot)$. For $\epsilon > 0$, consider the perturbed random variable $(X_\epsilon, Y_\epsilon) = (X, Y+\epsilon)$. Let the joint density of this perturbed random variable be $\tilde f_\epsilon$, and the marginal of $Y_\epsilon$ by $\tilde G_\epsilon$. Recall that for every convex function $\phi$ satisfying $\phi(1) = 0$, it is possible to define the divergence $D_\phi(p || q) = \int \phi\left(\frac{p(x)}{q(x)} \right)q(x) dx$ for two probability densities $p$ and $q$. Since such divergences satisfy the data-processing inequality, it is clear that 
\begin{equation}\label{eq: data_proc}
D_\phi(\tilde f_\epsilon || \tilde f) \geq D_\phi(\tilde G_\epsilon || \tilde G).
\end{equation}
Choosing $\phi(t) = (t-1)^p$, and using Taylor's expansion, it is easy to see that 
\begin{align*}
D_\phi(\tilde f_\epsilon || \tilde f) &=
\int_{\real^2} \frac{\abs{\tilde f(x,y)-\tilde f(x,y-\epsilon)}^p}{\tilde f(x,y)^{p-1}} dxdy 
= \epsilon^p \left(\int_{\real^2} \frac{\abs{\partial \tilde f(x,y)/\partial y}^p}{\tilde f(x,y)^{p-1}} dxdy \right)+ o(\epsilon^p).\\
\end{align*}
And similarly,
\begin{align*}
D_\phi(\tilde G_\epsilon || \tilde G) &=
\int_{\real} \frac{\abs{\tilde G(y)-\tilde G(y-\epsilon)}^p}{\tilde G(y)^{p-1}} dy 
= \epsilon^p \left(\int_{\real} \frac{\abs{d\tilde G(y)/dy}^p}{\tilde G(y)^{p-1}} dy \right)+ o(\epsilon^p).\\
\end{align*}
Substituting in inequality \eqref{eq: data_proc}, dividing by $\epsilon^p$, and taking the limit as $\epsilon \to 0$ yields
\begin{equation}
\int_{\real^2} \frac{\abs{\partial \tilde f(x,y)/\partial y}^p}{\tilde f(x,y)^{p-1}} dxdy \geq \int_{\real} \frac{\abs{d \tilde G/dy}^p}{\tilde G(y)^{p-1}} dy.
\end{equation}
The above inequality is exactly equivalent to that in Theorem \ref{thm: Car91} using the substitution $\tilde G = G^p$ and $\tilde f = f^p$. Although we focused on joint densities over $\real \times \real$, the same argument also goes through for random variables on $\real^m \times \real^n$. 

Recall that Definition \ref{def: l1_fisher} redefines the case of $p=1$ in Theorem \ref{thm: Car91}. Such redefinitions could indeed be done for $p>1$ as well. However, the perturbation argument presented above makes it clear that if $p > 1$, the $\phi$-divergence between a random variable (taking uniform values on some compact set) and its perturbation will be $+\infty$, since their respective supports are mismatched. Thus, analogous definitions for $p > 1$ will not yield useful bounds for such distributions. Using Definition \ref{def: l1_fisher}, we now establish superadditivity results for the $L_1$-Fisher information.

\begin{lemma}\label{lemma: l1_direction}
\label{thm: I_1(X)_1 < I_1(X_1)}
Let $X$ be an $\real^n$-valued random variable with a smooth density $f_X(\cdot)$.  Let $u \in \real^n$ be any unit vector. Define $X \cdot u$ to be the projection of $X$ along $u$. Then the following inequality holds when both sides are well-defined:
\begin{align}
I_1(X\cdot u) \leq I_1(X)_u.
\end{align}
\end{lemma}
\begin{proof}
Define the random variable $X_\epsilon \defn X + \epsilon u$. Then the distribution of $X_\epsilon$ satisfies
\begin{align}
f_{X_\epsilon}(x) = f_X(x - \epsilon u),
\end{align}
and is therefore a translation of $f_X$ along the direction $u$ by a distance $\epsilon$. Using the data-processing inequality for total-variation distance, we obtain 
\begin{align}\label{eq: tv}
d_{TV}(X\cdot u, X_\epsilon \cdot u) \leq d_{TV}(X_\epsilon, X), 
\end{align}
where $d_{TV}$ is the total variation divergence. Notice that $X_\epsilon \cdot u = X \cdot u + \epsilon$, and thus $f_{X_\epsilon \cdot u}(x) = f_{X \cdot u}(x-\epsilon)$. Dividing the left hand side of inequality \eqref{eq: tv} by $\epsilon$ and taking the limit as $\epsilon \to 0$, we obtain 
\begin{align*}
\lim_{\epsilon \to 0_+} \frac{d_{TV}(X\cdot u, X_\epsilon \cdot u) }{\epsilon} &= \frac{1}{2} \lim_{\epsilon \to 0_+} \int_\real \frac{\abs{f_{X\cdot u}(x) - f_{X \cdot u}(x - \epsilon)}}{\epsilon} dx\\
%
%
&\stackrel{(a)}= I_1(X \cdot u).
\end{align*}
Here, equality $(a)$ follows by the definition of $I_1(X \cdot u)$ and the assumption that it is well-defined. Doing a similar calculation for the right hand side of inequality \eqref{eq: tv} leads to 
\begin{align*}
\lim_{\epsilon \to 0_+} \frac{d_{TV}(X, X_\epsilon)}{\epsilon} &= \frac{1}{2} \lim_{\epsilon \to 0_+}  \int_{\real^n} \frac{\abs{f_X(x) - f_{X}(x-\epsilon u)}}{\epsilon} dx\\
&\stackrel{(a)}= \frac{1}{2} I_1(X)_u.
\end{align*}
The equality in $(a)$ follows from the definition of $I_1(X)_u$ and the assumption that it is well-defined.
\end{proof}
Our next result is a counterpart to the superadditivity property of Fisher information as in inequality \eqref{eq: l2_super}.

\begin{theorem}\label{thm: l1_super}
\label{cor: sum I_1(X_i) < I_1(X)}
Let $X = (X_1, \dots, X_n)$ be an $\real^n$-valued random variable. Then the following superadditivity property holds:
\begin{align*}
\sum_{i=1}^n I_1(X_i) \leq I_1(X).
\end{align*}
\end{theorem}
\begin{proof}
Applying Lemma \ref{lemma: l1_direction} for the unit vectors $e_1, \dots, e_n$, we obtain
\begin{align*}
\sum_{i=1}^n I_1(X_i) &\leq \sum_{i=1}^{n}I_1(X)_{e_i}
= I_1(X).
\end{align*}
\end{proof}

\subsection{Surface integral form of the $L_1$-Fisher information} 
If we consider a random variable $X$ that takes values uniformly over a set $K \subseteq \real^n$, then the $L_1$-Fisher information superaddivity from Theorem \ref{thm: l1_super} allows us to derive surface area inequalities once we observe two facts: 
\begin{enumerate}
\item[(a)] The  $L_1$-Fisher information $I_1(X)$ is well-defined for $X$ and is given by a surface integral over $\partial K$, and
\item[(b)] The quantity $I_1(X)_{e_i}$ may be calculated exactly given the sizes of all slices parallel to $e_i^\perp$, or may be lower-bounded by using any finite number of slices parallel to $e_i^\perp$.
\end{enumerate}

Establishing the surface integral result in part (a) requires making some assumptions on the shape of the geometric body. We focus on the class of polyconvex sets \cite{Sch14, KlaRot97}, which are defined as follows:
\begin{definition}
A set $K \subseteq \real^n$ is called a polyconvex set if it can be written as $K = \cup_{i=1}^m C_i$, where $m < \infty$ and each $C_i$ is a compact, convex set in $\real^n$ that has positive volume. Denote the set of polyconvex sets in $\real^n$ by $\cK$.
\end{definition}
In order to make our analysis tractable and rigorous, we first focus on polytopes and prove the polyconvex case by taking a limiting sequence of polytopes.
A precise definition of a polytope is as follows:
\begin{definition}
Define the set of polytopes, denoted by $\cP$ to be all subsets of $\real^n$ such that every $K \in \cP$ admits a representation $K = \cup_{j=1}^m P_j,$ where $m > 0$ and $P_j$ is a compact, convex polytope in $\real^n$ with positive volume for each $1 \leq j \leq m$.
\end{definition}



In what follows, we make observations $(a)$ and $(b)$ precise.

\begin{theorem}\label{thm: l1_surface}
Let $X$ be uniformly distributed over a  polytope $K$. Then the following equality holds:
\begin{equation}\label{eq: eii}
I_1(X) = \frac{1}{V_n(K)}\int_{\partial K} \norm{n(x)}_1 dS.
\end{equation}
\end{theorem}

\begin{proof}[Proof of Theorem \ref{thm: l1_surface}]
The equality in \eqref{eq: eii} is not hard to see intuitively. Consider the set $K$ and its perturbed version $K_\epsilon$ that is obtained by translating $K$ in the direction of $e_i$ by $\epsilon$. The $L_1$ distance between the uniform distributions on $K$ and $K_\epsilon$ is easily seen to be 
\begin{align*}
\frac{1}{V_n(K)} \left(V_n(K \cup K_\epsilon) - V_n(K \cap K_\epsilon)\right).
\end{align*}
\begin{figure}[ht]
			\centering
			\includegraphics[width=0.4\textwidth]{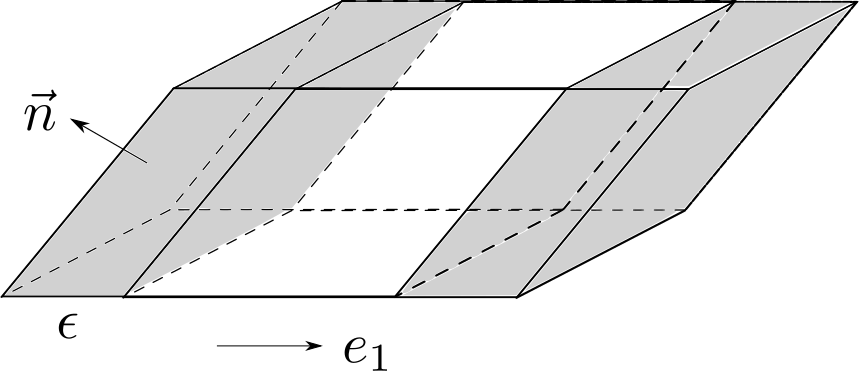}
			\caption{Perturbing a set by  $\epsilon$}
			\label{fig: surface projection}
		\end{figure}
As shown in Figure \ref{fig: surface projection}, each small patch $dS$ contributes $\abs{n(x) \cdot e_i} dS$ volume to $(K \cup K_\epsilon) \setminus (K \cap K_\epsilon)$, where $n(x)$ is the normal to the surface at $dS$. Summing up over all such patches $dS$ yields the desired conclusion. We make this proof rigorous with the aid of two lemmas:

\begin{lemma}[Proof in Appendix \ref{proof: lemma: I_1 in N_i}]
	\label{lemma: I_1 in N_i}
	Let $X$ be uniformly distributed over a compact measurable set $K \subseteq \mathbb{R}^n$.  If there exists an integer $L$ such that the intersection between $K$ and any straight line can be divided into at most $L$ disjoint closed intervals, then 
	\begin{equation}
		\label{eqn: I_1(X)_i in N_i}	
		I_1(X)_{e_i} = \int_{\mathbb{R}^{n-1}} \frac{2N_i(\ldots, \widehat{x_i},\ldots)}{V_n(K)}dx_1 \ldots \widehat{dx_i}\ldots dx_n.
	\end{equation}
	Here $\widehat{x_i}$ stands for removing $x_i$ from the expression.  The function $N_i(\ldots,\widehat{x_i},\ldots)$ is the number of disjoint closed invervals of the intersection of $K$ and line $\{ X_j = x_j, \: 1\le j \le n, j \neq i\}$.  
\end{lemma}
The above lemma does not require $K$ to be a polytope. However, the surface integral Lemma \ref{cor: I_1 as surface projection} below uses this assumption. 

\begin{lemma}[Proof in Appendix \ref{proof: cor: I_1 as surface projection}]
	\label{cor: I_1 as surface projection}
	Let $X$ be uniform over a polytope $K \in \cP$. Then
\[		
	\int_{\mathbb{R}^{n-1}} \frac{2N_i(\ldots, \widehat{x_i},\ldots)}{V_n(K)}dx_1 \ldots \widehat{dx_i}\ldots dx_n   = \frac{1}{V_n(K)}\int_{\partial K} |n(x) \cdot {e_i}| dS.
	\]
	Here  ${n}(x)$ is the normal vector at point $x \in\partial K$ and $dS$ is the element for surface area.  
\end{lemma}
Lemmas \ref{lemma: I_1 in N_i} and \ref{cor: I_1 as surface projection} immediately yield the desired conclusion, since $I_1(X) = \sum_{i=1}^n I_1(X)_{e_i}$
and $\norm{n(x)}_1 = \sum_{i=1}^n \abs{n(x) \cdot e_i}.$ 
\end{proof}

Our goal now is to connect $I_1(X_i)$ to the size of the slices of $K$ along $e_i^\perp$.  

\subsection{$L_1$-Fisher information via slices}

Consider the marginal density of $X_1$, which we denote by $f_{X_1}$. It is easy to see that for each $x_1 \in \supp f_{X_1}$, we have
\begin{equation}
f_{X_1}(x_1) = \frac{V_{n-1}(K \cap (e_1^\perp + x_1))}{V_n(K)}.
\end{equation}
Thus, the distribution of $X_1$ is determined by the slices of $K$ by hyperplanes parallel to $e_1^\perp$. Since Theorem \ref{thm: l1_super} is expressed in terms of $I_1(X_i)$, where each $X_i$ is a real-valued random variable, we establish a closed form expression for real-valued random variables in terms of their densities as follows:
\begin{lemma}[Proof in Appendix \ref{proof: thm: I_1 single variable}]
\label{thm: I_1 single variable}
Let $X$ be a continuous real-valued random variable with density $f_X$. If we can find $
 	-\infty=a_0< a_1 < \ldots < a_{M+1} = \infty$
 such that (a) $f_X$ is continuous and monotonic on each open interval $(a_i, a_{i+1})$; (b) For $i=0, \ldots ,M$, the limits
\begin{align*}
	f(a_i^+) &= \lim_{x \rightarrow a_i^+} f_X(x) \text{ for } i=1,\ldots,M, \quad \text{ and } \\
	f(a_i^-) &= \lim_{x \rightarrow a_i^-} f_X(x) \text{ for } i=1,\ldots,M
\end{align*}
exist and are finite. 
Then
	
\begin{equation}
	\label{eqn: I_1 as sum of limits}
		I_1(X) = \sum_{i=0}^M |f(a_{i+1}^-) - f(a_i^+)| + \sum_{i=1}^{M} |f(a_i^+) - f(a_i^-)|.
\end{equation}	
\end{lemma}

We can see that the first sum in \eqref{eqn: I_1 as sum of limits} captures the change of function values on each monotonic interval and the second term captures the difference of the one-sided  limits at end points. The following two corollaries are immediate.

\begin{corollary}
\label{cor: step function}
	Let $X$ be uniformly distributed on finitely many disjoint closed intervals; i.e., there exist disjoint intervals $[a_i, b_i] \subseteq \mathbb{R}$ for $i \in [N]$ and $\tau \in \mathbb{R}$
	such that
	\[
		f_X(x) = 
		\begin{cases}
			\tau & x \in \cup_{i=1}^{N} [a_i, b_i],  \quad \text { and }\\
			0 & \text{otherwise,}
		\end{cases}
	\]
	then $I_1(X) = 2N\tau.$
\end{corollary}

\begin{corollary}
\label{cor: I_1 for unimodal}
Let $X$ be a real-valued random variable with unimodal piecewise continuous density function $f_X$. Then the following equality holds:
\begin{align}
I_1(X) = 2\norm{f}_\infty.
\end{align}
\end{corollary}

Lemma \ref{thm: I_1 single variable} gives an explicit expression to compute $I_1$ when we know the whole profile of $f_X$.  When $f_X$ is only known for certain values $x$, we are able to establish a lower bound for $I_1(X)$. Note that knowing $f_X$ for only certain values corresponds to knowing the sizes of slices along a certain hyperplanes.

\begin{corollary}
	\label{cor: I_1 single var ineq}
Let $X \sim f_X$ where $f_X$ is as in Lemma \ref{thm: I_1 single variable}. If there exists a set  
	\[
		S = \{ -\infty = \theta_0 < \theta_1 < \ldots< \theta_N < \theta_{N+1} = \infty\}
	\]
	such that $f_X$ is continuous at each $\theta_i$ for $i \in [N]$, then
	
		\[
		I_1(X)  \ge \sum_{i=0}^N |f(\theta_{i+1}) - f(\theta_i)|.
	\]
\end{corollary}

\begin{proof}
We can find $T=\{a_i |i = 0, \ldots, M+1\}$ where $a_0 = \theta_0 = -\infty$, $a_{M+1} = \theta_{N+1} = +\infty$ such that they satisfy the conditions in Lemma \ref{thm: I_1 single variable}, and
	\begin{equation}
		I_1(X) = \sum_{i=0}^M |f_X(a_{i+1}^-) - f_X(a_i^+)| + \sum_{i=1}^{M} |f_X(a_i^+) - f_X(a_i^-)|.
		\label{eqn: I_1 continuous}
	\end{equation}
	Consider the set  $S \cup T = \{c_0 , \dots, c_{L+1}\}  $, which divides $\mathbb{R}$ into subintervals
	\[
		(c_i, c_{i+1}) \quad \text{ for } 0 \leq i \leq L+1.
	\]
We claim that 
\begin{equation}\label{eq: cc}
I_1(X) = \sum_{i=0}^L |f_X(c_{i+1}^-) - f_X(c_i^+)| + \sum_{i=1}^{L} |f_X(c_i^+) - f_X(c_i^-)|.	
\end{equation}
For the second term, note that
\begin{equation}
\sum_{i=1}^{L} |f_X(c_i^+) - f_X(c_i^-)| = \sum_{i=1}^{M} |f_X(a_i^+) - f_X(a_i^-)|,
\end{equation}
since $f_X$ is assumed to be continuous at $\theta_i$ for $i \in [N]$. The points in $S \setminus T$ subdivide each of the intervals $(a_i, a_{i+1})$;
i.e., for each interval $(a_i, a_{i+1})$ we can find an index $j_0$ such that $a_i = c_{j_0} < c_{j_0+1} < \dots < c_{j_0+ r} < c_{j_0+r+1} = a_{i+1}$, and the monotonicity of the function over $(a_i, a_{i+1})$ gives
\begin{equation}
\abs{f_X(a_i^+) - f_X(a_{i+1}^-)} = \sum_{j=0}^r \abs{f_X(c_{j_0+j+1}^-) -f_X(c_{j_0+j}^+)}.
\end{equation} 
Summing up over all intervals yields equality \eqref{eq: cc}. 
To conclude the proof, note that $f_X$ is not necessarily monotonic in the interval $(\theta_i, \theta_{i+1})$. Thus, if we have indices $\theta_i = c_{k_0}< \dots < c_{k_0+s+1} = \theta_{i+1}$, the triangle inequality yields
\begin{align*}
\abs{f_X(\theta_{i+1}) - f_X(\theta_i)}  &\stackrel{(a)}= \abs{f_X(\theta_{i+1}^-) - f_X(\theta_i^+)}\\
&= \Big| \sum_{u = 0}^s f_{X}(c_{k_0+u}^+) - f_{X}(c_{k_0+u+1}^-) + \sum_{u=1}^s f_X(c_{k_0+u}^-) - f_X(c_{k_0+u}^+)\Big|\\
&\leq \sum_{u = 0}^s \abs{f_{X}(c_{k_0+u}^+) - f_{X}(c_{k_0+u+1}^-)}+ \sum_{u=1}^s \abs{f_X(c_{k_0+u}^-) - f_X(c_{k_0+u}^+)}. 
\end{align*}
Here, equality $(a)$ follows from the continuity of $f_X$ at the points in $S$. Performing the above summation over all intervals $(\theta_i,\theta_{i+1})$ for $0\leq i \leq N$, and using equality \eqref{eq: cc}, we may conclude the inequality
	\[
		I_1(X) \ge \sum_{i=0}^N |f(\theta_i) - f(\theta_{i+1})|.
	\]
	\end{proof}
	\begin{remark}\label{remark: cont}
	Suppose $K$ is the union of two squares joined at the corner as shown in Figure \ref{fig: continuity}. Let $X$ be uniformly distributed on $K$. Suppose also that the slice of $K$ is known only at $\theta_1$. By direct calculation, we have $I_1(X \cdot e_1) = 2$, since $X \cdot e_1$ is uniform over $[0,1]$. Notice that $f_{X \cdot e_1}(\theta_1) = 2$, and thus the bound from Corollary \ref{cor: I_1 single var ineq} is $4$, which is \emph{larger} than $I_1(X \cdot e_1)$. This reversal is due to the discontinuity of $f_{X \cdot e_1}$ at the sampled location $\theta_1$---$f_{X \cdot e_1}(\theta_1)$ equals neither the left limit or the right limit at $\theta_1$. To avoid such scenarios, we require continuity of the density at sampled points. 
	\begin{figure}
	\begin{center}
	\includegraphics[scale = 0.5]{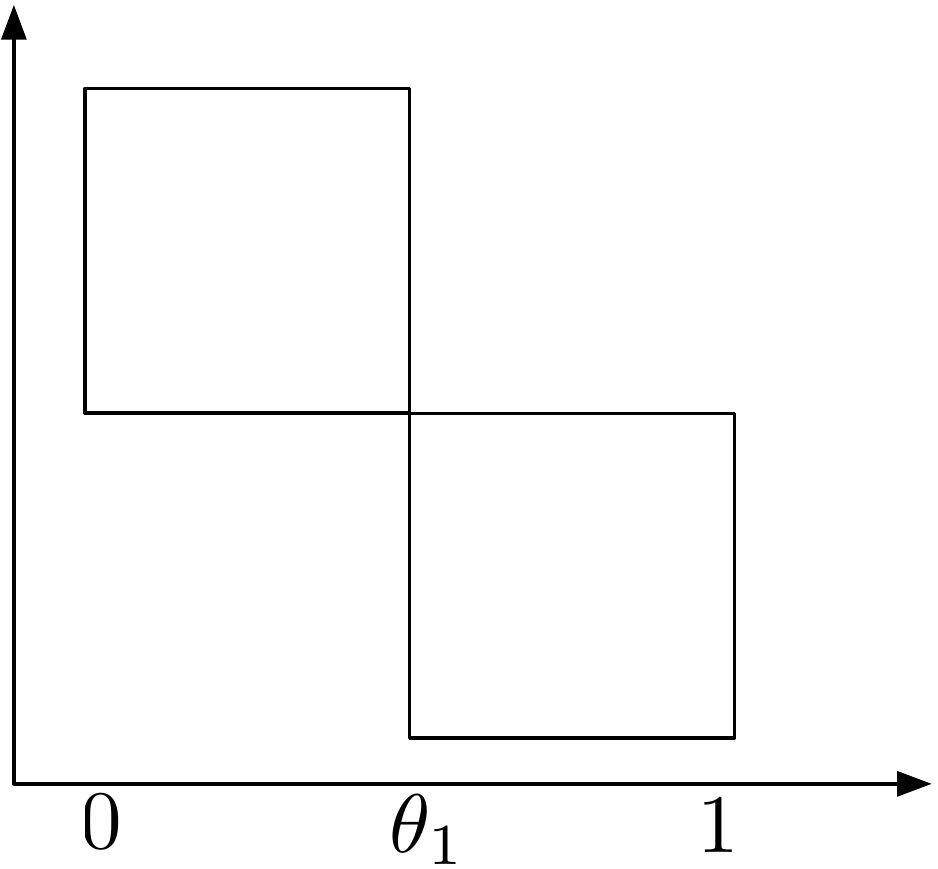}\caption{Uniform distribution over a union of squares}\label{fig: continuity}
	\end{center}
	\end{figure}

		\end{remark}

Corollary \ref{cor: I_1 single var ineq} shows that under mild conditions, we can estimate $I_1(X)$ when only limited information is known about its density function.

\subsection{Procedure to obtain lower bounds on the surface area} 
We first verify that the assumptions required by Lemma \ref{thm: I_1 single variable} are satisfied by the marginals of uniform densities over polytopes.
\begin{lemma}[Proof in Appendix \ref{proof: lemma: verify}]\label{lemma: verify}
Suppose $X = (X_1, \dots, X_n)$ is uniformly distributed over a polytope $K \in \cP$. Let $u$ be any unit vector and let $f_{X \cdot u}$ be the marginal density of $X \cdot u$. Then $f_{X\cdot u}(\cdot)$ satisfies in Lemma \ref{thm: I_1 single variable}.
\end{lemma}

Now suppose $X = (X_1, \dots, X_n)$ is uniformly distributed over a polytope $K$. Since $K$ is a polytope, we may write $K = \cup_{i=1}^m P_i$ where each $P_i$ is a compact, convex polytope.  Theorem \ref{thm: l1_super} provides the lower bound:
\begin{align*}
\frac{1}{V_n(K)}\int_{\partial K} \norm{n(x)}_1 dS \geq \sum_{i=1}^n I_1(X_i).
\end{align*}
To derive surface area bounds, notice that 
\begin{align*}
\sqrt n = \sqrt n \norm{n(x)}_2 \geq \norm{n(x)}_1,
\end{align*}
and thus
\begin{equation}
\frac{V_{n-1}(\partial K)}{V_n(K)} \geq \frac{1}{\sqrt n} \sum_{i=1}^n I_1(X_i).
\label{eqn: surface area > I_1}
\end{equation}

Suppose we know the sizes of some finite number of slices by hyperplanes parallel to $e_i^\perp$ for $i \in [n]$. We may use Corollary \ref{cor: I_1 single var ineq} to obtain lower bounds $\frac{B_i}{V_{n}(K)}$ on $I_1(X_i)$ for each $i \in [n]$ using the available slice information. This leads to the lower bound 
\begin{align*}
\frac{V_{n-1}(\partial K)}{V_n(K)} \geq \frac{1}{\sqrt n} \sum_{i=1}^n \frac{B_i}{V_{n}(K)},
\end{align*}
and thereby we may conclude the lower bound 
$$V_{n-1}(\partial K) \geq \frac{1}{\sqrt n} \sum_{i=1}^n {B_i}.$$ This is made rigorous in the following result, which may be considered to be our main result concerning surface areas.
\begin{theorem}\label{thm: surface_area}
Let $K$ be a polyconvex set. For $i \in [n]$, suppose that we have $M_i \geq 0$ slices of $K$ obtained by hyperplanes parallel to $e_i^\perp$, with sizes $\alpha^i_1, \dots, \alpha^i_{M_i}$. Then the surface area of $K$ is lower-bounded by
\begin{equation}\label{eq: poly}
V_{n-1}(\partial K) \geq \frac{1}{\sqrt n} \sum_{i=1}^n \left( \sum_{j=0}^{M_i} \abs{\alpha^i_j - \alpha^i_{j+1}}\right),
\end{equation}
where $\alpha^i_0, \alpha^i_{M_i+1} = 0$ for all $i \in [n]$.
\end{theorem}
\begin{proof}
Suppose that the $M_i$ hyperplanes parallel to $e_i^\perp$ are given by $(e_i^\perp + t_i^j)$ for $j \in [M_i]$. 
Let $K$ be a polyconvex set with a representation $K = \cup_{i=1}^m C_i$ where $C_i$ are compact, convex sets. For each $C_i$, we construct a sequence of convex polytopes $\{P^k_i\}$ which approximate $C_i$ from the outside. This means that $C_i \subseteq P^k_i$ for all $n \geq 1$ and $\lim_{k \to \infty} d(P^k_i, C_i) \to 0$, where $d$ is the Hausdorff metric. (This is easily achieved, for instance by sampling the support function of $C_i$ uniformly at random and constructing the corresponding polytope.) Consider the sequence of polytopes $P^k = \cup_{i=1}^m P^k_i$. For each $k$, we would like to assert that inequality \eqref{eq: poly} holds for the polytope $P^k$; i.e. we would like to lower bound $V_{n-1}(\partial P^k)$ using the slices of $P^k$ at $(e_i^\perp + t_i^j)$ for $i \in [n]$ and $j \in [M_i]$. The only difficulty in applying Corollary \ref{cor: I_1 single var ineq} to obtain such a lower bound on $V_{n-1}(\partial P^k)$ is the continuity assumption, which states that the marginal of the uniform density of $P_k$ on $e_i$, denoted by $f_{P^k \cdot e_i}$, should be be continuous at $t_i^j$ for all $i \in [n]$ and all $j \in [M_i]$. However, this is easily ensured by choosing an outer approximating polytope for $C_i$ that has no face parallel to $e_i^\perp$ for all $i \in [n]$. 

 To complete the proof for $K$, we need to show that $\lim_{k \to \infty} V_{n-1}(\partial P^k) = V_{n-1}(\partial K)$, and $\lim_{k \to \infty} V_{n-1}((e_i^\perp + t_i^j) \cap P^k) = V_{n-1}((e_i^\perp + t_i^j) \cap K)$ for any $i \in [n]$ and any $j \in [M_i]$. To show this, we use the following lemma \cite{MesSpo12}: 	

\begin{lemma}[Lemma 1 \cite{MesSpo12}]\label{lemma: intersect}
Let $K_1,\dots K_m \subseteq \real^n$ be compact sets. Let $\{K^k_i\}$, $k \geq 1$ be a sequence of
compact approximations converging to $K_i$ in Hausdorff distance, such that $K_i \subseteq K_i^n$ for
all $n \geq 1$ and for $i \in [m]$. Then it holds that
\begin{equation}
\lim_{k \to \infty} d\left(\cap_{i=1}^m K_i, \cap_{i=1}^m K^k_i \right) = 0.
\end{equation}
\end{lemma}
Using Lemma \ref{lemma: intersect}, we observe that for any collection of indices $1 \leq i_1 < \dots i_l \leq m$, 
we must have $d(P^k_{i_1} \cap \dots P^k_{i_l}, C_{i_1} \cap \dots \cap C_{i_l}) \to 0$ as $k \to \infty$. Since surface area is convex continuous with respect to the Hausdorff measure \cite{KlaRot97, Sch14}, we have the limit
\begin{equation}\label{eq: limit1}
\lim_{n\to \infty} V_{n-1}(\partial(P^k_{i_1} \cap \dots P^k_{i_l})) = V_{n-1}(\partial(C_{i_1} \cap \dots \cap C_{i_l})).
\end{equation} 
Moreover, surface area is a valuation on polyconvex sets \cite{KlaRot97, Sch14} and thus the surface area of a union of convex sets is obtained using the inclusion exclusion principle. In particular, the surface area of $K$ is
\begin{equation}\label{eq: limit2}
V_{n-1}(\partial K) = \sum_{i=1}^n V_{n-1}(\partial C_i) - \sum_{i_1 < i_2} V_{n-1}(\partial(C_{i_1} \cap C_{i_2})) + \dots + (-1)^{m+1} V_{n-1}(\partial(\cap_{i=1}^m  C_i)),
\end{equation}
and the surface area of $P^k$ is given by
\begin{equation}\label{eq: limit3}
V_{n-1}(\partial P^k) = \sum_{i=1}^n V_{n-1}(\partial P^k_i) - \sum_{i_1 < i_2} V_{n-1}(\partial(P^k_{i_1} \cap P^k_{i_2})) + \dots + (-1)^{m+1} V_{n-1}(\partial(\cap_{i=1}^m  P^k_i)).
\end{equation}
Using the limit in equation \eqref{eq: limit1}, we may conclude that every single term in \eqref{eq: limit3} converges to the corresponding term in \eqref{eq: limit2}, and so 
\begin{equation}
\lim_{k \to \infty} V_{n-1}(\partial P^k) = V_{n-1}(\partial K).
\end{equation}

We now show that each slice of $P^k$ converges in size to the corresponding slice of $K$. Let $H$ be some fixed hyperplane that is orthogonal to one of the coordinate axes. Since each $C_i$ can be replaced by a polytope $\cap_{k=1}^n  P^k_i$, we can assume without loss of generality that for each $i \in [m]$, the sequence of polytopes that approximate $C_i$ from outside is monotonically decreasing; i.e., $P^k_i \supseteq P^{k+1}_i$ for all $k \geq 1$.  For any fixed compact convex set $L \subseteq H$, Lemma \ref{lemma: intersect} yields
\begin{equation}\label{eq: LL}
d(L \cap P^k_i, L \cap C_i) \to 0, 
\end{equation}
and thus the $(n-1)$-dimensional volume of the two sets also converges. Picking $L$ to be $P^1_i \cap H$, we see that $L \cap P^k_i =  H  \cap P^k_i$, and $L \cap C_i$ is $H \cap C_i$, and thus equation \eqref{eq: LL} yields 
\begin{equation}
d(H  \cap P^k_i, H  \cap C_i) \to 0.
\end{equation}
The sequence of set $H \cap P^k_i$ for $n \geq 1$ is an outer approximation to $H \cap C_i$ that converges in the Hausdorff metric. Therefore, using Lemma \ref{lemma: intersect}, 
\begin{equation}
d((H \cap P^k_{i_1}) \cap \dots \cap (H \cap P^k_{i_l}), (H \cap C_{i_1}) \cap \dots \cap (H \cap C_{i_l})) \to 0.
\end{equation}
Using the continuity of the volume functional, 
\begin{equation}\label{eq: IE2}
V_{n-1}((H \cap P^k_{i_1}) \cap \dots \cap (H \cap P^k_{i_l})) \to V_{n-1}((H \cap C_{i_1}) \cap \dots \cap (H \cap C_{i_l})).
\end{equation}
Now an identical argument as above says that the $(n-1)$-dimensional volume of $H \cap K$ is obtained via an inclusion exclusion principle applied to the convex sets $H \cap C_i$ for $i \in [m]$. Applying equation \eqref{eq: IE2} to all the terms in the inclusion exclusion expression, we conclude that
\begin{equation}
V_{n-1}(H \cap P^k) \to V_{n-1}(H \cap K).
\end{equation}
This concludes the proof.
\end{proof}

Note that there is nothing restricting us to hyperplanes parallel to $e_i^\perp$. For example, suppose we have slice information available via hyperplanes parallel to $\{u_1^\perp, \dots, u_m^\perp\}$ for some unit vectors $u_i$ for $i \in [m]$. In this case, we have the inequality 
\begin{align*}
\frac{1}{V_n(K)}\int_{\partial K} \left(\sum_{j=1}^m \abs{n(x) \cdot u_j} \right) dS \geq \sum_{j=1}^m I_1(X \cdot u_j).
\end{align*}
Using the slice information, we may lower bound $I_1(X \cdot u_i)$ via Corollary \ref{cor: I_1 single var ineq}. Suppose this bound is $\frac{1}{V_n(K)}\sum_{j=1}^m B_j$. To arrive at a lower bound for the surface area, all we need is the best possible constant $C_n$ such that
\begin{align*}
C_n \geq \sum_{j=1}^m \abs{n(x) \cdot u_j}
\end{align*}
for all unit vectors $n(x)$. (This constant happened to be $\sqrt n$ when $u_j$'s were the coordinate vectors.) With such a constant, we may conclude
\begin{align*}
V_{n-1}(\partial K) \geq \frac{\sum_{j=1}^m B_j}{C_n}.
\end{align*}
In Appendix \ref{appendix: example}, we work out the surface area lower bound from Theorem \ref{thm: surface_area} for a particular example of a nonconvex (yet polyconvex) set.
\section{Conclusion}\label{section: end}

In this paper, we provided two different families of geometric inequalities to provide (a) Lower bounds on the volumes of convex sets using their slices, and (b) Lower bounds on the surface areas of polyconvex sets using their slices. These inequalities were derived using information theoretic tools. The volume bounds were obtained by using the Brascamp-Lieb subadditivity of entropy in conjunction with entropy bounds for log-concave random variables. Our main innovation in the surface area bounds is interpreting superadditivity of Fisher information as a consequence of the data-processing inequality applied to perturbed random variables. With this interpretation, we show that using the total variation distance for data-processing allows use to derive superadditivity results for the $L_1$-Fisher information. Crucially, the $L_1$-Fisher information is well-defined even for non-smooth densities, and thus we are able to calculate it for uniform distributions over compact sets.

There are a number of future directions worth pursuing. One interesting question is whether the volume bounds can be tightened further using entropy bounds for log-concave random variables that depend not just on the maximum value of the density, but also on the size of the support. Note that this means knowing the largest slices as well as the sizes of the projections of a convex set. Another interesting question is characterizing the equality cases of the superadditivity of Fisher information in Theorem \ref{thm: l1_super}, and thereby get a better understanding of when the resulting bounds provide meaningful estimates on the surface area of geometric body.

\bibliographystyle{unsrt}
\bibliography{Bibliography-Loomis_Whitney}

\begin{appendix}
\section{Proof of Lemma \ref{lemma: I_1 in N_i}}\label{proof: lemma: I_1 in N_i}

If $K \in \cP$, the assumption in Lemma \ref{lemma: I_1 in N_i} may be verified. Clearly, $K$ has finitely many faces $F_1, F_2, \dots, F_M$. For a line $\ell$ intersecting $K$ in some closed intervals, one of two events can happen. Either $\ell \cap F_j$ is one of the intervals, or the interval has endpoints that are marked by $\ell \cap F_{i_1}$ and $\ell \cap F_{i_2}$ for some $i_1, i_2 \in [M]$. The maximum number of intervals may be loosely bounded by $L \defn M+{M \choose 2}$, which is finite. 

	We show \eqref{eqn: I_1(X)_i in N_i} for $i=1$, and the others can be proved in the same way.  Since $X$ is uniformly distributed over $K$, 
	\[
		f_X(x) = \begin{cases}
			\frac{1}{V_n(K)} & \forall x \in K, \\
			0 & \text{otherwise.}
		\end{cases}
	\]
	Let 
	\[
		F(x_2, \ldots, x_n, \epsilon) = \int_{\mathbb{R}} \frac{|f_X(x_1, \ldots,x_n) - f_X(x_1 - \epsilon, \ldots,x_n)|}{\epsilon} dx_1.
	\]
	We claim that there exists $g(x_2, \ldots, x_n) \in L^1$, such that
	\[
		F(x_2, \ldots, x_n, \epsilon) \le g(x_2, \ldots, x_n).
	\]
	This would allow us to use the dominated convergence theorem to conclude 
	\[
		\lim_{\epsilon \rightarrow 0} \int_{\mathbb{R}^{n-1}} F(x_2, \ldots, x_n, \epsilon) dx_2 \ldots dx_n = \int_{\mathbb{R}^{n-1}} \lim_{\epsilon \rightarrow 0} F(x_2, \ldots, x_n, \epsilon) dx_2 \ldots dx_n
	\]
	Fix the coordinates $x_2, \ldots, x_n$.  If $(x_2, \ldots, x_n) \notin P_{e_1^\perp}(K)$, we clearly have $F(x_2, \ldots, x_n, \epsilon)= 0$.  Let $(x_2,\ldots,x_n) \in P_{e_1^\perp}(K)$.  Since $K$ intersects any straight line at most $L$ times, $f_X(x_1, \ldots, x_n)$ is a constant function on at most $L$ line segments and 0 else where.  We can write
it as $\sum_i f_i(x_1)$ where each $f_i$ is a constant function with value $\frac{1}{\vol_{n}(K)}$ on a small interval of $x_1$ and $0$ elsewhere.  Let $f_k$ be a function in this sum.  We consider $F(x_2,\ldots,x_n,\epsilon)$ in the following situations.
	\begin{enumerate}
		\item Support of $f_k$ is larger than or equal to $\epsilon$, then
		\begin{align*}
				 			 		\int_{\mathbb{R}} 
			 		 \frac{|f_k(x_1) - f_k(x_1-\epsilon)|}{\epsilon} dx 
			 		 = 2 \cdot \frac{1}{V_n(K)} \epsilon / \epsilon 
			 		= \frac{2}{V_n(K)}.
		\end{align*}
See Figure~\ref{fig: interval > eps}.
		\begin{figure}[ht]
			\centering
			\includegraphics[width=0.4\textwidth]{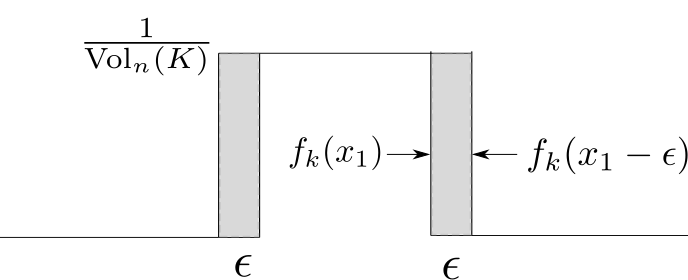}
			\caption{Support $> \epsilon$}
			\label{fig: interval > eps}
		\end{figure}
		\item Support of $f_k$ is $\epsilon'<\epsilon$, then
		\begin{align*}
\int_{\mathbb{R}} \frac{|f_k(x_1) - f_k(x_1-\epsilon)|}{\epsilon} dx_1 = 2 \cdot \frac{1}{V_n(K)} \epsilon' / \epsilon \le \frac{2}{V_n(K)}.
		\end{align*}
See Figure~\ref{fig: interval < eps}.
		\begin{figure}[ht]
			\centering
			\includegraphics[width=0.4 \textwidth]{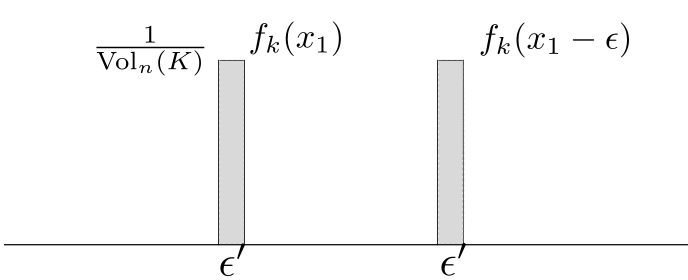}
			\caption{Support $< \epsilon$}
			\label{fig: interval < eps}
		\end{figure}

		In both cases, we have
		\[
				  \int_{\mathbb{R}} 
			 		 \frac{|f_k(x_1) - f_k(x_1-\epsilon)|}{\epsilon} dx_1 
			 		 \le \frac{2}{V_n(K)}.
		\]
	\end{enumerate}
	Therefore 
	\begin{align*}
\int_{\mathbb{R}} \frac{|f_X(x_1,\ldots,x_n) - f_X(x_1 - \epsilon,\ldots,x_n)|}{\epsilon} dx_1 
		 &= \int_{\mathbb{R}} \frac{|\sum_i f_i(x_1) - \sum_i f_i(x_1 - \epsilon)|}{\epsilon} dx_1 \\
		 &\le \int_{\mathbb{R}} \frac{\sum_i |f_i(x_1) - f_i(x_1 -\epsilon)|}{\epsilon} dx_1 \\
		 &\le \sum_i \frac{2}{V_n(K)} \\
		 &\le \frac{2L}{V_n(K)}.
	\end{align*}
	Let 
	\[
		g(x_2, \ldots, x_n) = \begin{cases}
			\frac{2L}{V_n(K)} & (x_2, \ldots, x_n) \in P_{e_1^\perp}(K), \\
			0 & \text{ otherwise. }
		\end{cases}
	\]
	Then $F(x_2, \ldots, x_n, \epsilon) \le g(x_2, \ldots, x_n)$, and
	\begin{align*}
		\int_{\mathbb{R}^{n-1}} g(x_2,\ldots,x_n) dx_2\ldots dx_n &= \int_{P_{e_1^\perp}(K)} \frac{2L}{V_n(K)}  dx_2 \ldots dx_n \\
		&= \frac{2L V_{n-1}(P_{e_1^\perp}(K))}{V_n(K)},
	\end{align*}
	which shows that $g$ is integrable. Using the dominated convergence theorem, we know that 
	\begin{align*}
		&\lim_{\epsilon \rightarrow 0^+} \int_{\mathbb{R}^n} \frac{|f_X(x_1,\ldots,x_n) - f_X(x_1 - \epsilon,\ldots,x_n)|}{\epsilon} dx_1dx_2\dots dx_n\\
		&= \lim_{\epsilon\to 0^+}\int_{\real^{n-1}}  F(x_2, \dots, x_n, \epsilon) dx_2 \dots dx_n\\
		&= \int_{\real^{n-1}} \lim_{\epsilon \to 0^+} F(x_2, \dots, x_n, \epsilon) dx_2 \dots dx_n\\
		& = \int_{\mathbb{R}^{n-1}} dx_2 \ldots dx_n\lim_{\epsilon \rightarrow 0^+} \int_{\mathbb{R}}\frac{|f_X(x_1,\ldots,x_n) - f_X(x_1 - \epsilon,\ldots,x_n)|}{\epsilon} dx_1. 
	\end{align*}
	Lastly, by Corollary \ref{cor: step function},
	 \[
		\lim_{\epsilon \rightarrow 0^+} \int_{\mathbb{R}}\frac{|f_X(x_1,\ldots,x_n) - f_X(x_1 - \epsilon,\ldots,x_n)|}{\epsilon} dx_1 = \frac{2N(x_2,\ldots,x_n)}{V_n(K)}.
	\]
	This concludes the proof.
\section{Proof of Lemma \ref{cor: I_1 as surface projection}}\label{proof: cor: I_1 as surface projection}
Recall that we need to show
\begin{equation}\label{eq: polytope_integral}
	\frac{1}{V_n(K)}\int_{\partial K} |n(x) \cdot {e_i}| dS = \frac{1}{V_n(K)} \int_{\mathbb{R}^{n-1}} 2N(x_2,\ldots,x_n) dx_2\ldots dx_n.
\end{equation}
Without loss of generality, let $i=1$. Denote $\partial K$ as $\cup_{j=1}^M F_j$ where $F_j$ are the faces of $K$. Let $n_j$ be the outward normal to $F_j$ for $j \in [m]$.  We have the equality
\begin{align*}
	V_{n-1}(P_{e_1^\perp}(F_i)) =  \int_{P_{e_1^\perp}(F_i)} dx_2dx_3\dots dx_n =  |n_i \cdot e_1| V_{n-1}(F_i). 
\end{align*}
Summing up for all $F_i$, the left hand side of \eqref{eq: polytope_integral} is given by
\begin{align*}
\frac{1}{V_n(K)}\int_{\partial K} |n(x) \cdot {e_1}| dS &= \frac{1}{V_n(K)}\sum_{i=1}^M \abs{n_i \cdot e_1} V_{n-1}(F_i)\\
&= \sum_{i=1}^M \frac{1}{V_n(K)}\int_{P_{e_1^\perp} (F_i)} dx_2dx_3 \dots dx_n.
\end{align*}
If for some $n_j$, the equality $n_j \cdot e_1 = 0$ holds, then $V_{n-1}(P_{e_1^\perp(F_j)}) = 0$.  Clearly,
\begin{align*}
 	\int_{P_{e_1^\perp}(F_i)} dx_2 dx_3 \ldots dx_n = \int_{P_{e_1^\perp}(F_i)} 2N(x_2, \ldots, x_n ) = 0.
 \end{align*} 
Without loss of generality, we assume $n_j \cdot e_1 \neq 0$. Let $\delta_{P_{e_1^\perp}(F_i)}(x_2,\ldots,x_n)$ be the indicator function on $P_{e_1^\perp}(F_i)$; i.e.,
\begin{align*}
	\delta_{P_{e_1^\perp}(F_i)}(x_2,\ldots,x_n) = \begin{cases}
		1 & (x_2,\ldots,x_n) \in P_{e_1^\perp}(F_i), \\
		0 & \text{otherwise.}
	\end{cases}
\end{align*}
Then  
\begin{align*}
\sum_{i=1}^M \int_{P_{e_1^\perp}(F_i)} dx_2\ldots dx_n 
	=& \sum_{i=1}^{M} \int_{\mathbb{R}^{n-1}} \delta_{P_{e_1^\perp}(F_i)}(x_2,\ldots,x_n) dx_2 \ldots dx_n \\
	=& \int_{\mathbb{R}^{n-1}} \sum_{i=1}^{M} \delta_{P_{e_1^\perp}(F_i)}(x_2,\ldots,x_n) dx_2 \ldots dx_n.
\end{align*}
For every $(x_2,\ldots,x_n)$, there will be $2N(x_2,\ldots,x_n)$ many $F_i$'s such that $(x_2,\ldots,x_n) \in F_i$. Therefore
\begin{align*}
	 \int_{\mathbb{R}^{n-1}} \sum_{i=1}^{k} \delta_{P_{e_1^\perp}(U_i)}(x_2,\ldots,x_n) dx_2 \ldots dx_n 
	= \int_{\mathbb{R}^{n-1}} 2N(x_2, \ldots, x_n) dx_2 \ldots dx_n,
\end{align*}
which completes the proof.

\section{Proof of Lemma \ref{thm: I_1 single variable}}\label{proof: thm: I_1 single variable}

We claim that 
	\begin{align}
		\lim_{\epsilon \rightarrow 0^+} \int_{a_i + \epsilon}^{a_{i+1}} \frac{|f_X(x) - f_X(x- \epsilon)|}{\epsilon} dx &= |f_X({a_{i+1}}^-) - f_X(a_i^+)| , \quad (i=0, \ldots, M) \label{eqn: a_i + eps -> a_i+1}, \text{    and }\\
		\lim_{\epsilon \rightarrow 0^+} \int_{a_i}^{a_i + \epsilon} \frac{|f_X(x) - f_X(x- \epsilon)|}{\epsilon}  dx &= |f_X(a_i^+) - f_X(a_i^-)|, \quad (i=1, \ldots, M). \label{eqn: a_i -> a_i + eps}
	\end{align}
	If $f_X$ is increasing on $(a_i,{a_{i+1}})$, then
	\begin{align*}
		\lim_{\epsilon \rightarrow 0^+} \int_{a_i+\epsilon}^{a_{i+1}} \frac{|f_X(x) - f_X(x- \epsilon)|}{\epsilon} dx &= \lim_{\epsilon \rightarrow 0^+}  \int_{a_i + \epsilon}^{a_{i+1}} \frac{f_X(x) - f_X(x-\epsilon)}{\epsilon} dx \\
		&= \lim_{\epsilon \rightarrow 0^+}  \int_{a_i + \epsilon}^{a_{i+1}} \frac{f_X(x)}{\epsilon} dx -  \int_{a_i + \epsilon}^{a_{i+1}} \frac{f_X(x-\epsilon)}{\epsilon} dx \\
		&= \lim_{\epsilon \rightarrow 0^+} \int_{a_i+\epsilon}^{a_{i+1}} \frac{f_X(x)}{\epsilon} dx + \int_{a_i}^{{a_{i+1}}-\epsilon} \frac{f_X(x)}{\epsilon} dx\\
		&= \lim_{\epsilon \rightarrow 0^+} - \int_{a_i}^{a_i+\epsilon} \frac{f_X(x)}{\epsilon} dx + \int_{{a_{i+1}} - \epsilon}^{a_{i+1}} \frac{f_X(x)}{\epsilon}  dx\\
		&= -f_X(a_i^+) + f_X({a_{i+1}}^-).
	\end{align*}
The last equality is true since $ \int_{a_i}^{a_i+\epsilon} \frac{f_X(x)}{\epsilon} = f_X(\theta)$ for $a_i <\theta_i < a_i + \epsilon$ due to mean value theorem.  This value approaches $f_X(a_i^+)$ when $\epsilon \rightarrow 0^+$.  Using the same argument, we can show $ \lim_{\epsilon \rightarrow 0^+}\int_{{a_{i+1}}- \epsilon}^{a_{i+1}} \frac{f_X(x)}{\epsilon} dx = f_X({a_{i+1}}^-)$ .	Similarly, when $f_X$ is decreasing on $(a,{a_{i+1}})$, we have 
	\[
		\lim_{\epsilon \rightarrow 0^+} \int_{a_i+\epsilon}^{a_{i+1}} \frac{|f_X(x) - f_X(x- \epsilon)|}{\epsilon} dx = |f_X(a_i^+) - f_X({a_{i+1}}^-)|
	\]
Therefore we have established \eqref{eqn: a_i + eps -> a_i+1}.
Similarly, $\int_{a_i}^{a_i+\epsilon} \frac{|f_X(x)- f_X(x- \epsilon)|}{\epsilon} = |f_X(\theta') - f_X(\theta'- \epsilon)|$ for $ a<\theta' < a + \epsilon $.  Since $a_i - \epsilon< \theta' - \epsilon< a_i $, this approaches $|f_X(a_i^-) - f_X(a_i^+)$ as $\epsilon \rightarrow 0^+$.  So we have also established \eqref{eqn: a_i -> a_i + eps}.
	Lastly,
	\begin{align*}
		I_1(X) &= 	 \lim_{\epsilon \rightarrow 0^+} 
		 		\int_{\mathbb{R}} \frac{|f_X(x) - f_X(x-\epsilon)|}{\epsilon} dx \\
		&= \lim_{\epsilon \rightarrow 0^+} \sum_{i=0}^{M} \int_{a_i + \epsilon}^{a_{i+1}}\frac{|f_X(x) - f_X(x- \epsilon)|}{\epsilon} dx + \sum_{i=1}^{M} \int_{a_i}^{a_i + \epsilon} \frac{|f_X(x) - f_X(x-\epsilon)|}{\epsilon} dx \\
		&= \sum_{i=0}^M |f_X(a_{i+1}^-) - f_X(a_i^+)| + \sum_{i=1}^{M} |f_X(a_i^+) - f_X(a_i^-)|.
	\end{align*}

\section{Proof of Lemma \ref{lemma: verify}}\label{proof: lemma: verify}

Without loss of generality, assume $u = e_1$. Let $K = \cup_{i=1}^m P_i$ where $P_i$ are compact, convex polytopes. Denote the projection of $f_X(\cdot)$ restricted to some set $C \subseteq K$ on the $e_1$ axis by $f_{C \cdot e_1}(\cdot)$. If $C$ is a convex polytope, we may verify that $f_{C \cdot e_1}$ is log concave, and therefore a continuous function on some closed interval. Furthermore, if $C$ is a convex and compact polytope, then we may triangulate $C$; i.e., express $C = \cup_{i=1}^r T_i$ where $T_i$ are $n$-dimensional compact simplices for $i \in [r]$ such that their interiors partition the interior of $C$. Then $f_{C \cdot e_1} = \sum_{i=1}^r f_{T_i \cdot e_1}$. Each function in the summation is a degree $(n-1)$ polynomial with a compact interval as its support in $\real$ \cite{Las15}. Thus, $f_{C \cdot e_1}$ is a continuous function consisting of finitely many pieces such that $f_{C \cdot e_1}$ restricted to each piece is a polynomial of degree $(n-1)$. Note that the overall density $f_{K \cdot e_1} \defn f_{X\cdot e_1}(\cdot)$ is given via the inclusion exclusion principle by
\begin{align*}
f_{K \cdot e_1} = \sum_{i=1}^m f_{P_i} - \sum_{i_1<i_2} f_{P_{i_1} \cap P_{i_2}} + \sum_{i_1<i_2<i_3} f_{P_{i_1} \cap P_{i_2} \cap P_{i_3}} + \dots
\end{align*}
For each collection of indices $i_1, \dots, i_k$, we have that $\cap_{j=1}^k P_{i_j}$ is a compact, convex polytope, possibly with 0 volume, but such sets do not contribute to the above sum so we only consider  cases where the intersection has a positive volume. The sum (or difference) of finitely many bounded continuous functions on closed intervals is easily to satisfy the following property: We may find finitely many points $-\infty =  \gamma_0 < \gamma_1 < \dots < \gamma_R < \gamma_{R+1} = + \infty$ such that the function is continuous on each open interval $(\gamma_i, \gamma_{i+1})$, and the left and right limits at the endpoints in each interval are finite. To verify the assumptions in Lemma \ref{thm: I_1 single variable}, we simply check that on each interval $(\gamma_i, \gamma_{i+1})$, the function does not have infinitely local optima. This is clearly true since restricted to $(\gamma_i, \gamma_{i+1})$, the function is a piecewise polynomial of degree $(n-1)$. This proves the claim. 

\begin{remark}
Note that in general, it is possible for the difference of log concave functions to have infinitely many local optima.  For example, if $f_1, f_2: [-1,1] \to \real$ such that $f_1(x) = 2-x^2$ and $f_2(x) = (2-x^2) + \frac{e^{-1/x^2}\sin(1/x)}{2}$, then both functions are concave and positive, and therefore log-concave. However, $f_2-f_1 = \frac{e^{-1/x^2}\sin(1/x)}{2}$ has infinitely many local optima close to 0.  The observation that the marginals in our case are piecewise polynomials is therefore necessary in the above argument.
\end{remark}

\section{Example}\label{appendix: example}

\begin{figure}[ht]
	\centering
	\includegraphics[width=0.4 \textwidth]{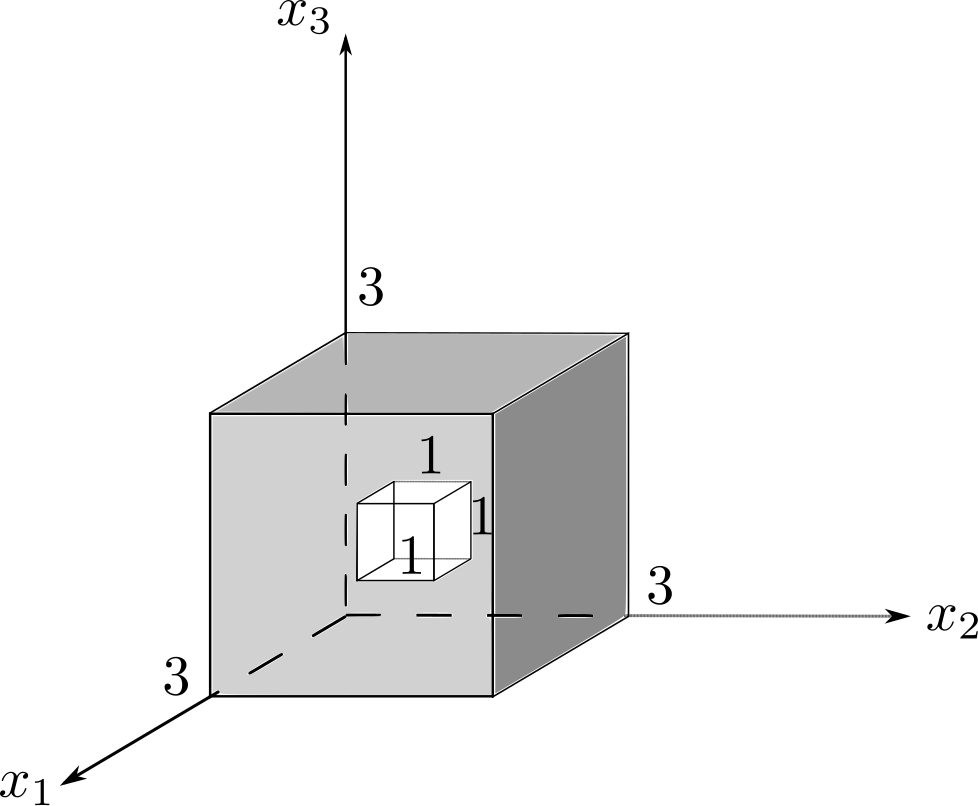}
	\caption{Cube with a hole}
	\label{fig: cube with a hole}
\end{figure}

To illustrate \eqref{eqn: surface area > I_1}, we consider the example of a cube with a hole in it as in Figure \ref{fig: cube with a hole}. Note that this is a nonconvex set, but is easily seen to be polyconvex. The density $f_{X_1}(x_1)$ can be computed using area of slices along $x_1$-axis, and
\begin{align*}
 	I_1(X_1) &= |f_{X_1}(0) - f_{X_1}(-\infty)| + 
 				|f_{X_1}(1) - f_{X_1}(0)| +
 				|f_{X_1}(3) - f_{X_1}(2)| +
 				|f_{X_1}(\infty) - f_{X_1}(3)| \\
 			&=  | \frac{9}{26} - 0 | + 
 				| \frac{8}{26} -\frac{9}{26}| + 
 				| \frac{9}{26} - \frac{8}{26}| + 
 				|0 - \frac{9}{26}| \\
 			&= \frac{20}{26}.
 \end{align*} 
 By symmetry,
 \begin{align*}
 	\frac{1}{\sqrt{3}} (I_1(X_1) + I_1(X_2) + I_1(X_3)) = \frac{60}{26\sqrt{3}} \approx 1.33,
 \end{align*}
By direct calculation, 
\begin{equation*}
	\frac{V_{n-1}(\partial K)}{V_n(K)} = \frac{48}{26} \approx 1.85.
\end{equation*}

\end{appendix}


\end{document}